\documentclass[11pt,a4paper]{article}

\usepackage{times}
\usepackage{soul}
\usepackage{url}
\usepackage[hidelinks]{hyperref}
\usepackage[utf8]{inputenc}
\usepackage[small]{caption}
\usepackage{graphicx}
\usepackage{amsmath}
\usepackage{amsthm}
\usepackage{booktabs}
\usepackage{algorithm}
\usepackage{algorithmic}
\title{Local Aggregation in Preference Games}% with Local Aggregation}

\author{
Angelo Fanelli\thanks{CNRS, (UMR-6211), France. Email: \texttt{angelo.fanelli@unicaen.fr}.}
\and
Dimitris Fotakis\thanks{
National Technical University of Athens, Greece. Email: \texttt{fotakis@cs.ntua.gr}.} 
}

\date{}

\usepackage{fancybox}
\usepackage{fullpage}
\usepackage{color}
\usepackage{hyperref}
\usepackage{amssymb, amsthm, mathtools}
\usepackage{amsmath}
\usepackage{array}
\usepackage{dutchcal}
\usepackage{todonotes}

\def\aggregation{\mathrm{aggr}} %aggregation function

\def\Fmedian{\mathrm{med}} % Fmedian aggregation

\def\Fmean{\mathrm{mean}} %  Fmean aggregation
\def\zz{\vec{z}} %generic state

\def\xx{\vec{x}} %other generic state
\def\x{x} %other generic strategy

\def\yy{\vec{y}} %other generic state
\def\y{y} %other generic strategy

\def\ee{\vec{e}} %equilibrium
\def\e{e}

\def\oo{\vec{o}} %optimum

\def\ss{\vec{s}} %internal opinions
\def\s{s}

\def\d{d} %distance between strategies

\def\D{D} % set of players with different strategies in the equilibrium and the optimum 

\def\ds{\pi} %distance between states

\def\w{w} %weight

\def\STUBB{\alpha} %stubbornness factor

\def\SI{\delta} % social impact

\def\STR{\tau} % stretch

\def\BOU{\beta} % dominance

\def\N{N} % set of players
\def\U{U} % universe
\def\S{S} % internal opinions
\def\Z{Z} % strategy space
\def\E{E} % set of equilibria
\def\OSUM{O_\mathrm{\textsc{Sum}}} % set of optimal states for SUM
\def\OMAX{O_\mathrm{\textsc{Max}}} % set of optimal states for Max

\def\SUM{\mathrm{\textsc{Sum}}} %social cost SUM
\def\MAX{\mathrm{\textsc{Max}}} %social cost MAX

\def\POASUM{\mathrm{\textsc{PoA}}^\mathrm{\textsc{Sum}}} %price of anarchy SUM
\def\POAMAX{\mathrm{\textsc{PoA}}^\mathrm{\textsc{Max}}} %price of anarchy MAX

\def\aggr{\mathrm{aggr}}
\def\reals{\mathbb{R}}

\newtheorem{theorem}{Theorem}
\newtheorem{lemma}{Lemma}
\newtheorem{proposition}{Proposition}

\begin{document}

\maketitle

\begin{abstract}
In this work we introduce a new model of decision-making by agents in a social network. Agents have innate preferences over the strategies but, because of the social interactions, the decision of the agents are not only affected by their innate preferences but also by the decision taken by their social neighbors. 
We assume that the strategies of the agents are embedded in an {approximate} metric space. 
Furthermore, departing from the previous literature, we assume that, due to the lack of information, each agent locally represents the trend of the network through an aggregate value, which can be interpreted as the output of an aggregation function.
We answer some fundamental questions related to the existence and efficiency of pure Nash equilibria.
\end{abstract}
%This allows to express the innate preferences as a function of the distance from a favourite strategy. 
%and each agent has a favourite strategy.  
%Agents innately prefer strategies closer to their favourite strategies to ones that are further away. 

%%%%%%%%%%%%%%%%%%%%%_INTRO_%%%%%%%%%%%%%%%%%%%%%%%%%%

\section{Introduction}
\label{sec:intro}

%Introductory paragraph
A significant volume of recent research investigates opinion formation and preference aggregation in social networks. The goal is to obtain a deeper understanding of how social influence affects decision making, and to which extent the tension between innate preferences and social influence could cause instability or inefficiency. The prevalent models are \emph{opinion formation games} \cite{BindelKO15} and \emph{discrete preference games} \cite{ChierichettiKO18}, whose principles go back to the classical works of \cite{Degroot} and \cite{Friedkin} on opinion formation. Though seemingly simple, both models are natural and expressible, have apparent similarities and some crucial differences.

% Gentle introduction and general context behind these opinion formation and discrete preference games
In a nutshell, both opinion formation and discrete preference games assume a finite population of $n$ agents and an underlying strategy space $Z$, which is the $[0, 1]$ interval for opinion formation games and a finite set with at least two strategies for discrete preference games (this is one of the key differences of the two models). Each agent $i$ has a fixed preferred strategy $s_i \in Z$ and adopts a strategy in $Z$, which represents a compromise between her preferred strategy and the strategies of her social neighbors. Both models need to formally define (i) how to quantify preferences over strategies; and (ii) how much each agent is influenced by her social neighbors. 

For the former, both models assume a distance function $d : Z \times Z \to \reals_{\geq 0}$ which quantifies dissimilarity between strategies. In opinion formation, $d$ is $L_2^2$ (i.e., the square of the $L_2$ norm, motivated by repeated averaging in \cite{Degroot,Friedkin}), while in discrete preference games, $d$ is any metric distance function (motivated by applications more general than opinion formation, as explained in \cite{ChierichettiKO18}). 

Both models assume that the influence exercised by agent's $j$ strategy on the strategy selection of agent $i$ is quantified by the \emph{influence weight} $w_{ij} \in [0, 1]$. Previous work often distinguishes between the symmetric case, where $w_{ij} = w_{ji}$, and the asymmetric case, where $w_{ij}$ and $w_{ji}$ may be different. Moreover, the confidence of each agent $i$ on her preferred strategy $s_i$ is quantified by the \emph{self-confidence} $\alpha_i \in [0, 1]$. %(often referred to as the \emph{self-confidence} of agent $i$). 
Another (less important) difference between opinion formation and discrete preference models is that most previous work on the latter (see e.g., \cite{ChierichettiKO18,AulettaCFGP16}) usually assumes the same self-confidence $\alpha \in [0, 1]$ for all agents $i$ and that each $\w_{ij}$ is either $1$ or $0$. %the same influence on $i$ from all her social neighbors $j$ (i.e., $w_{ij} = 1$, if $j$ is connected to $i$, and $w_{ij} = 0$, otherwise). 

% -- however, \cite{LolakapuriBNPD19} considers discrete preference games with general self-confidence levels $w_i$ and influence weights $w_{ij}$). 

The discussion above is elegantly summarized by the cost function of agent $i$ in a strategy profile $\vec{z}$: 
\[
 \bar{c}_i(\vec{z}) = \alpha d(s_i, \vec{z}(i)) + (1-\alpha) \sum_{j \neq i} w_{ij} d(\vec{z}(i), \vec{z}(j))\,.
\]
Naturally, each agent $i$ selects her strategy $\vec{z}(i)$ so as to minimize this cost. 

In the last few years, there has been considerable interest in equilibrium properties (e.g., existence, computational complexity, convergence, price of anarchy and stability) of opinion formation and discrete preference games (and several variants and generalizations of them). The main message is that opinion formation games are well-behaved due to their single-dimensional continuous strategy space, while discrete preference games exhibit more complex behavior. Specifically, opinion formation games admit a unique equilibrium which can be computed efficiently and has small price of anarchy in the symmetric case (see e.g., \cite{BindelKO15,BGM13,GS12}). In discrete preference games, an equilibrium with low price of stability exists, but there might exist multiple equilibrium points, some with unbounded price of anarchy \cite{ChierichettiKO18}; moreover, finding an equilibrium can be computationally hard \cite{LolakapuriBNPD19}, while strategies at equilibrium may be significantly different from the agents' preferred strategies. 

%\paragraph{Motivation.}
\medskip\noindent{\bf Motivation.}
An important assumption underlying virtually all results above is that each agent is fully aware of the strategies of her social neighbors and responds optimally to them. This assumption is crucial for establishing uniqueness and convergence to equilibrium in opinion formation (see e.g., \cite{GS12}) and forms the basis for the price of anarchy and stability results (see e.g., \cite{BindelKO15,BGM13,BiloFM18,ChierichettiKO18}). However, having access to the strategies of all one's neighbors is questionable in real life and in modern large online social networks, just because getting to know these strategies requires a large amount of information exchange. For (a somewhat extreme) example, imagine that one does not crystallize her opinion on a topic before she has extensively discussed it with all her acquaintances! 

Therefore, the assumption that agents form their strategies based on explicit (and complete) knowledge of the strategies of all their neighbors has been criticized, especially in the context of opinion formation games and their best response dynamics. Researchers have studied opinion formation with restricted information regimes \cite{UL04} and opinion dynamics where agents learn the strategy of a single random neighbor in each round \cite{FPS16,FKKS18}. The goal is to understand the extent to which limited information about the strategies of one's social neighbors affects the properties of equilibria in opinion formation. 

In this work, we take a different, and to the best of our knowledge, novel approach towards addressing the same research question. Instead of focusing on best response dynamics and assuming that the strategies of a small random set of neighbors are available in each round, we assume that each agent has access to a single representative strategy, which can be regarded as the output of an \emph{aggregation function} that condenses the strategies of her neighbors into a single one. This aggregate strategy could be provided by the network (in case of online social networks, e.g., Facebook or Twitter could provide a summary of the main posts of one's friends on a specific topic) or could be estimated by a poll. To further motivate our approach, we may think of traditional political voting, where voters have fixed innate preferences over the candidates, but polls (which is a form of preference aggregation) might cause the voters to change their vote (see also \cite{EpitropouFHS19}, where the opinion formation model involves an estimation of the average opinion of all agents).

%Instead of a traditional political voting, we may also consider an informal voting process, through an online voting tool, in which voters are free to revise their votes as a response to the latest observed outcome obtained through a predefined aggregation rule (see also \cite{EpitropouFHS19}, where the opinion formation model involves an estimation of the average opinion of all agents). 

%\paragraph{Preference Games with Local Aggregation: The Model}
\medskip\noindent{\bf Preference Games with Local Aggregation: Our Model.}
We introduce a very general game-theoretic model of decision-making by agents in a social network, which we refer to as \emph{preference games with local aggregation} (or simply \emph{preference games}, for brevity). Similarly to opinion formation and discrete preference games, each agent selects a strategy trying to be faithful as much as possible to her preferred strategy (which is immune to the choices made by the other agents) and, at the same time, to blend with the environment. But in our model, the environment is summarized by the aggregate strategy of her social neighbors. 

The basic formal setting of preference games with local aggregation is inspired by that of discrete preference games (but there are some crucial differences, as we explain below). As in previous work, we consider a strategy space $Z$ which is the same for all agents, but in preference games with local aggregation, $Z$ may be either discrete or continuous. Another important difference from previous work is that we allow preferred strategies not to belong to $Z$. %(e.g., $Z$ may be discrete, while preferred strategies may belong to the convex hull of $Z$). 
This allows the agents to have more elaborate preferred strategies and to account for preferences that possibly cannot be fully disclosed in public. 
So, we assume a strategy ``universe'' $U$ and that $Z \subseteq U$ (e.g., $Z$ may be discrete, while $U$ may be the convex hull of $Z$). We say that the game is \emph{restricted}, if the preferred strategy $s_i \in U\setminus Z$ %$s_i \not\in Z$ 
for all agents $i$, \emph{unrestricted}, if $s_i \in Z$ for all agents $i$, and \emph{semi-restricted}, otherwise. 

As in discrete preference games, we assume a %metric 
distance function $d : U \times U \to \reals_{\geq 0}$, which quantifies dissimilarity between strategies. But, instead of an exact metric, we let $d$ be an approximate metric %distance function, 
so as to also allow for other natural dissimilarity functions, such as $L_2^2$. 
Moreover, we assume the same self-confidence level $\alpha \in [0, 1]$ for all agents and that influence weights $w_{ij}$ are not necessarily symmetric, have $w_{ii} = 0$ and are normalized so that $\sum_{j} w_{ij} = 1$, i.e., the total influence exercised to any agent $i$ sums up to $1$ (in \cite{ChierichettiKO18}, influence weights are symmetric and not normalized). 

The major new key ingredient is an \emph{aggregation function} $\aggr$ which, for each agent $i$, maps the strategies $\vec{z}_{-i} = (z_1, \ldots, z_{i-1}, z_{i+1}, \ldots, z_n)$ of all agents other than $i$ and the influence weights $\vec{w}_i = (w_{ij})_{j\in \N}$ on $i$ to an aggregated strategy $\aggr(\vec{z}_{-i}, \vec{w}_i) \in Z$ that ``summarizes'' the strategies in $\vec{z}_{-i}$. We usually write $\aggr_i(\vec{z}_{-i})$, instead of $\aggr(\vec{z}_{-i}, \vec{w}_i)$, for brevity. Typical aggregation functions are the \emph{mean} (the best response function in opinion formation) and the \emph{median} (the best response function in discrete preference games). However, most of our results hold for general aggregation functions that satisfy \emph{unanimity} (i.e., if $\vec{x}_{-i} = (x, \ldots, x)$, then $\aggr_i(\vec{x}_{-i}) = x$) and \emph{consistency} (i.e., if $\vec{x}_{-i}$ and $\vec{y}_{-i}$ satisfy $\sum_{j \neq i} w_{ij} d(\vec{x}(j), \vec{y}(j)) = 0$, then $\aggr_i(\vec{x}_{-i}) = \aggr_i(\vec{y}_{-i})$). We refer to aggregation functions that satisfy unanimity and consistency as \emph{feasible}. 

In preference games with local aggregation, the cost of agent $i$ in a strategy profile $\vec{z}$ is: 
\begin{equation}\label{eq:cost}
 c_i(\vec{z}) = \alpha d(s_i, \vec{z}(i)) + (1-\alpha) d(\vec{z}(i), \aggr_i(\vec{z}_{-i}))\,.
\end{equation}
Namely, the cost of agent $i$ is a convex combination of her \emph{innate cost} $d(s_i, \vec{z}(i))$ and her \emph{disagreement cost} $d(\vec{z}(i), \aggr_i(\vec{z}_{-i}))$. 
%The former is the cost of adopting a strategy $\vec{z}(i)$ dissimilar to $i$'s innate preference $s_i$, and the latter is the cost of adopting a strategy $\vec{z}(i)$ dissimilar to a strategy aggregating the strategies in her social circle. 
As usual, each agent $i$ selects her strategy $\vec{z}(i)$ so as to minimize $c_i(\vec{z})$. The crucial difference from opinion formation and discrete preference models is that strategy selection of agent $i$ solely depends on $\aggr_i(\vec{z}_{-i})$, which is a single strategy in $\Z$, and not on the entire strategy vector $\vec{z}_{-i}$. 
An interesting direction for future research is to assume that $\aggr_i(\vec{z}_{-i})$ may belong to $\U \setminus \Z$.  %and does not necessarily belong to the strategy space $\Z$.

%\paragraph{Contribution.}
\medskip\noindent{\bf Contribution.}
The conceptual contribution of our work is the new model of preference games with local aggregation. The model is mostly inspired by discrete preference games, but is quite general and allows for a systematic study and a new perspective on the fundamental question of how much limited information about the preferences of one's social neighbors affects the equilibrium properties in opinion formation and preference aggregation in social networks. 
On the technical side, we provide a comprehensive set of general results on the existence and the structure of equilibria and on the price of anarchy of preference games with local aggregation. Our results hold for {any approximate metric} distance function $d$ and any feasible aggregation function $\aggr$. A general message of our results is that low self-confidence levels (i.e., $\alpha \leq 1/2$) help with the existence of equilibria and simplify their structure, while high self-confidence levels (i.e., $\alpha > 1/2$) help with the price of anarchy. Moreover, our price of anarchy analysis for $\alpha > 1/2$ implies novel bounds on the price of anarchy of certain variants of opinion formation and discrete preference games.

Specifically, in Section \ref{sec:existence} we show that if $\alpha \leq 1/2$, \emph{consensus} (i.e., a state where all agents adopt the same strategy) is a pure Nash equilibrium of preference games with local aggregation (Theorem~\ref{thm:existence:<}); 
if $\alpha \geq 1/2$, we show that the state in which each agent $i$ adopts her preferred strategy $s_i$ is a pure Nash equilibrium, and moreover such equilibrium is unique when $\alpha > 1/2$ (Theorem \ref{thm:existence:>}). %existence of pure Nash equilibria only for unrestricted games, (i.e., when $s_i \in Z$ for all agents $i$),  where the unique equilibrium is when each agent $i$ adopts her preferred strategy $s_i$.
{The above two results hold under the more stringent  assumption that $d$ is an exact metric.} 
 Existence of pure Nash equilibrium for restricted games with $\alpha > 1/2$ requires more assumptions (e.g., on the strategy space, the aggregation function, or the distance function) and is an interesting direction for further research.

In Section \ref{sec:PoA} we consider the price of anarchy with respect to the total cost and the maximum cost of the agents. We observe that  if $\alpha = 0$, the cost of every agent at equilibrium is $0$ and therefore the price of anarchy is $1$. On the other hand, if $\alpha \in (0, 1/2]$, the price of anarchy of restricted and unrestricted games can be unbounded for both objectives (Proposition \ref{prop:poa:(0,1/2]} and Proposition \ref{prop:poa:(0,1/2]:b}). So, if self-confidence level is low, the price of anarchy of preference games with local aggregation behaves similarly to that of discrete preference games \cite{ChierichettiKO18} and of opinion formation games with binary strategies \cite{FerraioliGV16}. 
Interestingly, if $\alpha > 1/2$, we show that the price of anarchy is bounded from above by $1+\frac{1-\alpha}{\alpha} \frac{\delta \tau}{\beta}$, for the total cost (Theorem~\ref{thm:poa:sum}), and by $1+\frac{1-\alpha}{\alpha} \frac{\tau}{\beta}$, for the maximum cost (Theorem~\ref{thm:poa:max}). 
We believe that all the three parameters $\delta$, $\beta$ and $\tau$ used in these bounds are intuitive and of interest; all of them are formally defined in Section~\ref{sec:model}. 
$\delta$ is the maximum \emph{social impact} $\sum_{j \neq i} w_{ji}$ of an agent $i$. Since $\sum_{j \neq i} w_{ij} = 1$, $\delta$ quantifies the asymmetry between the influence received and exercised by any agent in the social network. $\beta$ is the so-called maximum \emph{boundary} of any agent. The boundary $\beta_i$ of agent $i$ quantifies how much closer a strategy $z$ can be to $i$'s preferred strategy $s_i$ compared against an equilibrium strategy of $i$. We believe that $\beta$ is a natural parameter and can be exploited in the proof of price of anarchy/stability bounds for other generalizations of discrete preference games. 
The most interesting parameter is the \emph{stretch} $\tau$, which directly quantifies how much we lose, in terms of equilibrium efficiency, because we only have access to an aggregate of the strategies in one's social neighborhood. %(see Section~\ref{sec:model} for a formal definition).
 To better demonstrate this point, let's assume that the aggregation function is the (weighted) median. Then, in a standard discrete preference game, the best response of an agent $i$ is the weighted median of $i$'s preferred strategy $s_i$ and $i$'s social neighbor strategies in $\vec{z}_{-i}$. In a preference game with local aggregation, on the other hand, $i$ receives the weighted median $\aggr_i(\vec{z}_{-i})$ and computes her best response as a weighted median of $s_i$ and $\aggr_i(\vec{z}_{-i})$. Communication efficient as it is, the latter may lead to more inefficient equilibria, since a small change in $\vec{z}_{-i}$ might cause a significant change in $\aggr_i(\vec{z}_{-i})$. The extent to which this can happen is quantified by $\tau$. 

In Section~\ref{sec:tau_beta}, we provide upper bounds on $\tau$ and $\beta$ ($\delta$ is always upper bounded by $n-1$), which implies that the price of anarchy of preference games with local aggregation is always bounded if $\alpha > 1/2$. Interestingly, our bounds on $\beta$ depend only on $\alpha$, while our bounds on $\tau$ may depend on the metric space, the influence weights and the aggregation function. 

In Section~\ref{sec:voting}, we introduce a specific preference game, which is motivated by $k$-approval voting and generalizes opinion formation with binary strategies \cite{FerraioliGV16}. The strategy space $Z$ consists of all binary strings of length $m$ with $k$ ones, preferred strategies lie in the convex hull of $Z$, local aggregation is the weighted median, and the distance function is $L_2^2$. %(which satisfies the triangle inequality only approximately in $U$, but it is equivalent to the Hamming distance, when restricted to $Z$).
Since $L_2^2$ is a $2$-approximate metric in $U$ (but it is equivalent to the Hamming distance when restricted to $Z$), results from Section \ref{sec:PoA} carry over this special case. 
The main result of this section is an upper bound on $\beta$ (Theorem \ref{thm:boundary:voting}), which implies an interesting upper bound on the price of anarchy for all $\alpha > 1/2$. An intriguing direction for further research is to determine under which assumptions $k$-approval voting game admits pure Nash equilibria for $\alpha > 1/2$. 

%Due to lack of space, we defer most of the proofs to the full version of this work. 

%\paragraph{Other Related Work.}
\medskip\noindent{\bf Other Related Work.}
Discrete preference games were introduced in \cite{ChierichettiKO18}, where it was shown that they are potential games, that the price of anarchy can be unbounded and that the price of stability is at most $2$. Moreover, the properties of discrete preference games for richer metrics, such as tree metrics, were studied. Recently, \cite{LolakapuriBNPD19} proved that computing a pure Nash equilibrium of discrete preference games is PLS-complete, even in a very restricted setting. Discrete preference games were generalized in \cite{AulettaCFGP16} and consistency between preferred strategies and equilibrium strategies were considered in \cite{AulettaCFGP17}. 

Our $k$-approval voting model bears some resemblance to opinion formation with binary preferences \cite{FerraioliGV16}. %Our analysis implies that the price of anarchy for that model is at most $\delta / (2\alpha-1)^2$ for $\alpha > 1/2$. 
%\cite{FerraioliGV16} 
They proved that the price of anarchy is unbounded for $\alpha \leq 1/2$. Our analysis complements theirs by showing that the price of anarchy for $\alpha > 1/2$ is at most $\delta / (2\alpha-1)^2$. %for $\alpha > 1/2$.
%\cite{FerraioliGV16} proved that the price of anarchy is unbounded for $\alpha \in [0, 1/2]$. 

Aggregating preferences under some metric %dissimilarity 
function has been received attention in algorithms (see e.g., \cite{AilonCN08}) and in social choice (see e.g., \cite{AnshelevichBEPS18}). 
%To the best of our knowledge, 
Ours is the first work where aggregation under metric dissimilarity functions is used in modelling decision-making by agents in a social network. %preference and opinion formation in social networks. 

%%%%%%%%%%%%%%%%%%%%%_MODEL_%%%%%%%%%%%%%%%%%%%%%%%%%%

\section{Notation, Definitions and Preliminaries} 
\label{sec:model}

Most of the notation and the model definition are introduced in Section~\ref{sec:intro}. Next, we introduce some additional notation and discuss some important preliminaries. 

We recall that $\N=\{1, \ldots, n\}$, with $n \geq 2$, is the set of agents, $U$ is the strategy universe, and $Z \subseteq U$, with $|Z| \geq 2$, is the strategy space of the agents. $S = \{ s_1, \ldots, s_n \} \subseteq U$ is the set of the agents' preferred strategies, and $N_S = \{ i \in N: s_i \in Z\}$ is the set of agents with preferred strategy in $Z$. So, $N_S = N$ denotes an unrestricted game, while $N_S = \emptyset$ denotes a restricted one. We recall that $w_{ij} \in [0, 1]$ is the amount of influence agent $j$ imposes on agent $i$. Unless stated otherwise, $w_{ij}$ may be different than $w_{ji}$. We always assume that $w_{ii} = 0$ and that $\sum_{j=1}^n w_{ij} = 1$. We recall that $\alpha \in [0, 1]$ is the confidence-level of the agents. We say that the agents are \emph{stubborn}, if $\alpha \in (1/2, 1]$, and \emph{compliant}, otherwise.

We refer to any $\vec{z} \in Z^n$ as a \emph{state} of the game. For any state $\zz$ and any strategy $z$, we let $[\zz_{-i}, z]$ be the new state obtained from $\zz$ by replacing its $i$-component $\vec{z}(i)$ with $z$ and keeping the remaining components unchanged. If $\vec{z} = (z, \ldots, z)$, %i.e., all agents have the same strategy $z$ in $\vec{z}$, 
we say that $\vec{z}$ is a \emph{consensus} on $z$.  

%We say that $\d$ is a \emph{$\rho$-approximate metric}, for some $\rho > 1$, if it satisfies (i), (ii) and $d(x, y) \leq \rho(d(x, z) + d(z, y))$, for all $x, y, z \in U$. 

We say that $\d: \U \times \U \mapsto \reals_{\geq 0}$ is a \emph{$\rho$-approximate metric}, for some $\rho \geq 1$, if it satisfies (i) $d(x, x) = 0$, for all $x \in U$; (ii) symmetry, i.e., $d(x, y) = d(y, x)$, for all $x, y \in U$; and (iii) (approximate) triangle inequality, i.e., $d(x, y) \leq \rho(d(x, z) + d(z, y))$, for all $x, y, z \in U$. 
%A metric space $\metric{\U, \d}$ is \emph{uniform} if $d(x, y) = 1$ for all $x \neq y$. 
We say that $\d$ is an exact metric (or simply metric) if $\rho = 1$. 
We say that $\d$ is \emph{uniform} if it is an exact metric such that $d(x, y) = 1$ for all $x \neq y$.
We assume that $d$ is a $\rho$-approximate metric, for some $\rho \geq 1$, unless stated otherwise.
%We assume that $d$ is a (general) metric on $U$, unless stated otherwise. 

%\paragraph{Feasible aggregation functions.}
\medskip\noindent{\bf Aggregation Functions.}
We consider aggregation functions that satisfy (i) \emph{unanimity}, i.e., if $\vec{x}_{-i}$ is a consensus on $x$, then $\aggr_i(\vec{x}_{-i}) = x$; and (ii) \emph{consistency}, i.e., for all $\vec{x}_{-i}$, $\vec{y}_{-i}$ with $\sum_{j \neq i} w_{ij} d(\vec{x}(j), \vec{y}(j)) = 0$, $\aggr_i(\vec{x}_{-i}) = \aggr_i(\vec{y}_{-i})$. We say that such aggregation functions are \emph{feasible}. 
In this work we focus on {feasible} aggregation functions. %that satisfy \emph{unanimity} and \emph{consistency}. 
When $\d$ is an exact metric on $\Z$, notable examples of feasible aggregation functions are the \emph{Fr\'{e}chet mean} and the \emph{Fr\'{e}chet median}.

Given any state $\zz$, the Fr\'{e}chet mean of agent $i$ in $\zz$, denoted by  $\Fmean_i(\zz_{-i})$, is any strategy in $\Z$ that minimizes the weighted sum of its squared distances to the strategies in $\zz_{-i}$. Formally, 
\begin{equation}\label{def:mean}
	\Fmean_i(\zz_{-i}) \in  {\arg\min}_{\y\in \Z} \sum_{j \neq i}\w_{ij}\d^2\big(\y, \zz(j)\big)\,.
\end{equation}
The Fr\'{e}chet median of agent $i$ in $\zz$, denoted by  $\Fmedian_i(\zz_{-i})$, is any strategy that minimizes the weighted sum of its distances to the strategies in $\zz_{-i}$\,:
\begin{equation}\label{def:median}
	\Fmedian_i(\zz_{-i}) \in  {\arg\min}_{\y\in \Z} \sum_{j\neq i}\w_{ij}\d\big(\y, \zz(j)\big)\,.
\end{equation}
We can show that the Fr\'{e}chet mean and the Fr\'{e}chet median are indeed feasible aggregation functions. 

The following proposition shows that both aggregation functions are feasible.

\begin{proposition}\label{prop:aggr:feasibility}
The Fr\'{e}chet mean and the Fr\'{e}chet median are feasible aggregation rules. 
\end{proposition}

\begin{proof}
	We prove the statement for the Fr\'{e}chet mean. 
	A virtually identical argument applies to the Fr\'{e}chet median.
	
	Let $i$ be any agent, and $\xx$ and $\yy$ be any pair of states.
	It is straightforward to verify unanimity, namely, that if $\xx_{-i}$ is a consensus on $\x\in\Z$, then $\Fmean(\xx_{-i}) = \x$. 
	In fact, every term of the summation would be $0$ in $\x$.
	
We proceed to prove consistency. Let us assume that $\sum_{j \neq i} w_{ij} d(\vec{x}(j), \vec{y}(j)) = 0$. If $\vec{x}(j) = \vec{y}(j)$ for all coordinates $j \neq i$, then $\xx_{-i}$ and $\yy_{-i}$ are identical and $\Fmean_{i}(\xx_{-i}) =\Fmean_{i}(\yy_{-i})$. So, let us assume that for some coordinates $j$, $\vec{x}(j) \neq \vec{y}(j)$. Since $\sum_{j \neq i} w_{ij} d(\vec{x}(j), \vec{y}(j)) = 0$, it must be $w_{ij} = 0$ for all coordinates $j$ with $\vec{x}(j) \neq \vec{y}(j)$. Equivalently, for every $j$ with $\w_{ij} > 0$, we have that $\xx(j) = \yy(j)$. 
	Therefore, for every $\y \in \Z$, it holds that
\begin{align*}
	\sum_{j\in \N\setminus\{i\}}\w_{ij}\d^2\big(\y, \xx(j)\big) 
	& = 
	\sum_{
	\begin{subarray}{c} 
		j\in \N\setminus\{i\} \\  \w_{ij} > 0
	\end{subarray}
	}
	\w_{ij}\d^2\big(\y, \xx(j)\big) %\\
	%& = 
	=
	\sum_{
	\begin{subarray}{c} 
		j\in \N\setminus\{i\} \\  \w_{ij} > 0
	\end{subarray}
	}
	\w_{ij}\d^2\big(\y, \yy(j)\big) %\\
	%& = 
	=
	\sum_{j\in \N\setminus\{i\}}\w_{ij}\d^2\big(\y, \yy(j)\big),
\end{align*}
	which implies that $\Fmean_{i}(\xx_{-i}) =\Fmean_{i}(\yy_{-i})$.
\end{proof}

%\paragraph{Pure Nash equilibria, social optima and price of anarchy.} 
\medskip\noindent{\bf Pure Nash Equilibria, Social Optima and Price of Anarchy.} 
A \emph{pure Nash equilibrium} (or \emph{equilibrium}, for brevity) is a state $\vec{z}$ such that for every agent $i$ and every strategy $z \in Z$, $c_i(\zz) \leq c_i([\zz_{-i}, z])$. $\E \subseteq \Z^n$ denotes the set of all pure Nash equilibria of a given preference game. A strategy $z^\ast \in \Z$ is a \emph{best response} of agent $i$ to a state $\zz$, if $z^\ast \in \arg\min_{z\in \Z} c_i([\zz_{-i}, z])$. %We say that a strategy $\x\in \Z$ is a \emph{dominant strategy} of agent $i$, if $c_i\big([\xx_{-i},\x]\big) \leq c_i\big([\xx_{-i},\y]\big)$, for all states $\xx_{-i}$ and $\y \in \Z$.
We say that a strategy $\x\in \Z$ is a \emph{(strictly) dominant strategy} of agent $i$, if $c_i\big([\xx_{-i},\x]\big) \leq c_i\big([\xx_{-i},\y]\big)$ ($<$, if strictly), for all states $\xx_{-i}$ and strategies $\y \in \Z$.

We measure the efficiency of each state $\zz$ according to a \emph{social objective}. We consider two social objectives, the \emph{social cost} $\SUM(\zz) = \sum_{i\in \N}c_i(\zz)$ and the \emph{maximum cost} $\MAX(\zz) = \max_{i\in \N} c_i(\zz)$. A state $\oo$ is \emph{optimal wrt. $\SUM$},  if $\SUM(\oo) \leq \SUM(\zz)$, for all states $\zz$. We denote by $\OSUM \subseteq \Z^n$ the set of optimal states wrt. $\SUM$, i.e., $\OSUM = \arg\min_{\zz \in \Z^n}\SUM(\zz)$. %Optimal states wrt. $\MAX$ and $\OMAX$ are defined similarly. 
Similarly, a state $\oo$ is \emph{optimal wrt. $\MAX$}, if $\MAX(\oo) \leq \MAX(\zz)$, for all states $\zz$. We let $\OMAX = \arg\min_{\zz \in \Z^n}\MAX(\zz)$ be the set of all optimal states wrt. $\MAX$. 

The \emph{price of anarchy} of a game wrt. $\SUM$ is $\POASUM = \max_{\ee \in \E}\frac{\SUM(\ee)}{\SUM(\oo)}$, if $\SUM(\oo) > 0$ for some state $\oo\in \OSUM$. If $\SUM(\oo) = 0$, then $\POASUM = +\infty$, if $\E \neq \OSUM$, and $\POASUM = 1$, if $\E = \OSUM$. The definition of $\POAMAX$ is similar.

%\paragraph{Equivalence and relative distance between states.}
\medskip\noindent{\bf Equivalence and Relative Distance between States.}
For every pair of states $\xx$, $\yy$, $\D(\xx,\yy) = \{j\in \N : \xx(j) \neq \yy(j)\}$ is the set of agents with different strategies in $\xx$ and $\yy$. If $\D(\xx,\yy) = \emptyset$, we say that $\xx$ and $\yy$ are \emph{globally equivalent}.
For every agent $i$ and all pairs of states $\xx$, $\yy$, we define the \emph{relative distance} of $\xx$ and $\yy$ for $i$ as $\ds_i(\xx,\yy) = \sum_{j\neq i}\w_{ij}\d\big(\xx(j),\yy(j)\big)$.
If  $\ds_i(\xx,\yy) = 0$, $\xx$ and $\yy$ are \emph{equivalent} for $i$. We observe that $\D(\xx,\yy) = \emptyset$ implies $\ds_i(\xx,\yy) = 0$ (while converse may not be true). Moreover, $\ds_i\big([\xx_{-i}, x],[\xx_{-i},y]\big) = 0$,  for all strategies $\x$ and $\y$. 
%Finally, as a result of property (b), we have that the symmetry extends also to $\pi_i$, i.e.,  $\ds_i(\xx,\yy) = \ds_i(\yy,\xx)$.

%In the reminder of this work we focus only on the case $\STUBB > 1$.
%As we will see in the following sections,  beside the stubbornness of the agents, the inefficiency of the equilibria are influenced by three additional non-negative parameters:  the \emph{social impact} of the agents, the \emph{stretch factor} and the \emph{boundary}. We continue our discussion  with the formal definition of these parameters; before that we introduce some notation about to denote the relations between states. 

%\paragraph{Social impact, stretch factor and  boundary.} 
\medskip\noindent{\bf Social Impact, Stretch and Boundary.}
The \emph{social impact} of agent $i$ is $\SI_i = \sum_{j\in \N}\w_{ji}$ and quantifies the intensity by which $i$ influences the environment. The \emph{global social impact} is $\SI = \max_{j \in N} \SI_j$. We observe that $\SI \in [1, n-1]$ and $\sum_{i \in N} \SI_i = n$. We refer to the case where influence weights are symmetric, i.e., $w_{ij} = w_{ji}$ for all $i$ and $j$, and $\SI = 1$ as the \emph{fully symmetric} case. 

The \emph{stretch} $\tau_i$ of agent $i$ quantifies how sensitive the aggregation function is wrt. changes in the state of the game. Formally, we let 
\[ \hat{\tau}_i =\!\! \inf_{\begin{subarray}{c}
                    \vec{e} \in \E\\
                    \vec{y} \in Z^n\\
                    \ds_i(\vec{e}, \vec{y}) > 0
                  \end{subarray}}
\!\!\!\!\left\{ r \geq 0\,\Big|\,\frac{d(\aggr_{i}(\vec{e}_{-i}),\aggr_{i}(\vec{y}_{-i}))}{\ds_i(\vec{e}, \vec{y})} \leq r \right\}. \]
We define $\tau_i = \max\{\hat{\tau_i}, 1\}$, so as to account for the case where $\ds_i(\vec{e}, \vec{y}) = 0$ and $d(\aggr_{i}(\vec{e}_{-i}),\aggr_{i}(\vec{y}_{-i})) = 0$ (recall the definition of feasible aggregation functions). Therefore, existence of $r$ is always guaranteed and $\tau_i$ is well defined. 
Since $\tau_i$ is used in the price of anarchy bounds, we can restrict its definition to optimal states $\vec{o}$, instead of arbitrary states $\vec{y}$. We use this more restricted definition in the proofs of propositions~\ref{prop:strech:o-consensus}~and~\ref{prop:strech:e-consensus}, in Section~\ref{sec:tau_beta}. 

The \emph{(global) stretch} is $\tau = \max_{j \in N} \tau_j$\,. At the conceptual level, the stretch quantifies how much the price of anarchy increases because agents only have access to an aggregate of the strategies in $\vec{z}$, instead of $\vec{z}$ itself. 

The \emph{boundary} of agent $i$, denoted by $\beta_i$ quantifies how much closer a strategy $x$ can be to $s_i$ compared against an equilibrium strategy $\vec{e}(i)$ of $i$. Formally, 
\begin{equation}\label{def:boundary} 
\beta_i = \min_{\vec{e} \in E, x \neq \vec{e}(i)} \frac{d(x, s_i)}{d(x, \vec{e}(i))}.
\end{equation}
For nontrivial games, there always exists a strategy $x$ with $d(x, \vec{e}(i)) > 0$. Thus, $\beta_i$ is well defined. The \emph{(global) boundary} is $\beta = \min_{j \in N} \beta_j$\,.

%%%%%%%%%%%%%%%%%%%%%_Existence_%%%%%%%%%%%%%%%%%%%%%%%%%%

\section{Existence of Pure Nash Equilibria}\label{sec:existence}

We first characterize the best responses of compliant and stubborn agents. The proof of Lemma~\ref{lem:dominant-strategy} is analogous to that of Lemma~\ref{lem:BR}. 
{We remark that the results in this section hold only under the  assumption that $\d$ is an exact metric.} 

\begin{lemma}\label{lem:BR}
If $\STUBB \leq 1/2$ (resp. $\STUBB < 1/2$), $\aggregation_i(\xx_{-i})$ is a (resp. the unique) best response of agent $i$ to $\xx_{-i}$, for every state $\xx$. 
\end{lemma}
\begin{proof}
Let us assume $\STUBB \leq 1/2$.
Let $i$ be any agent and $\xx$ be any state. 
Let $\x$ be any strategy in $\Z\setminus\{\aggregation_i(\xx_{-i})\}$ (recall that $|\Z| \geq 2$). Since $\aggregation_i(\xx_{-i}) \in Z$, we have
	\begin{align} 
	c_i\big([\xx_{-i},\aggregation_i(\xx_{-i})]\big) &= \STUBB\d\big(\aggregation_i(\xx_{-i}), \s_i\big) + (1-\STUBB)\d\big(\aggregation_i(\xx_{-i}), \aggregation_i(\xx_{-i})\big) \nonumber\\
	&= \STUBB\d\big(\aggregation_i(\xx_{-i}), \s_i\big) \nonumber\\  
	&\leq  \STUBB\d\big(\x, \s_i\big) + \STUBB\d\big(\x, \aggregation_i(\xx_{-i})\big) \label{br:eq:step1}\\
	&\leq  \STUBB\d\big(\x, \s_i\big) + (1 - \STUBB)\d\big(\x, \aggregation_i(\xx_{-i})\big) \label{br:eq:step2}\\ 
	&=   c_i\big([\xx_{-i}, \x]\big), \nonumber
	\end{align}
	where  \eqref{br:eq:step1} follows from the triangle inequality and \eqref{br:eq:step2} from $\STUBB \leq 1/2$, which implies that $\STUBB \leq (1 - \STUBB)$. If $\STUBB < 1/2$, \eqref{br:eq:step2} is strict. Hence, $\aggregation_i(\xx_{-i})$ is the unique best response of $i$.
\end{proof}

%\begin{proof}
%Let $i$ be any agent, $\xx$ any state, and $\x$ any strategy in $\Z\setminus\{\aggregation_i(\xx_{-i})\}$. Using the facts that $\aggregation_i(\xx_{-i}) \in Z$ and $\d(\aggregation_i(\xx_{-i}), \aggregation_i(\xx_{-i})) = 0$, we obtain that:
%%
%	\begin{align} 
%	&\quad c_i\big([\xx_{-i},\aggregation_i(\xx_{-i})]\big) \nonumber\\ 
%	&= \STUBB\d\big(\aggregation_i(\xx_{-i}), \s_i\big) \nonumber\\
%	&\leq  \STUBB\d\big(\x, \s_i\big) + \STUBB\d\big(\x, \aggregation_i(\xx_{-i})\big)\label{br:eq:step1}\\
%	&\leq  \STUBB\d\big(\x, \s_i\big) + (1 - \STUBB)\d\big(\x, \aggregation_i(\xx_{-i})\big) \label{br:eq:step2}\\ 
%	&=   c_i\big([\xx_{-i}, \x]\big), \nonumber
%	\end{align}
%%
%	where  \eqref{br:eq:step1} follows from the triangle inequality and \eqref{br:eq:step2} from $\STUBB \leq 1/2$. If $\STUBB < 1/2$, \eqref{br:eq:step2} is strict. 
%\end{proof}

\begin{lemma}\label{lem:dominant-strategy}
If $\STUBB \geq 1/2$ (resp. $\STUBB > 1/2$), $\s_i$ is a (resp. the unique) dominant strategy of any agent $i \in N_S$. 
\end{lemma}
\begin{proof}
Let us assume $\STUBB \geq 1/2$. Let $i$ be any agent in $\N_S$ and $\xx$ be any state with $\xx(i) \neq s_i$ (recall that $|\Z|\geq 2$). Since $s_i \in Z$, we have  
	\begin{align} 
	%&\quad c_i\big([\xx_{-i},\s_i]\big) \nonumber \\  
	c_i\big([\xx_{-i},\s_i]\big) &= \STUBB\d\big(\s_i, \s_i\big) + (1 - \STUBB)\d\big(\s_i, \aggregation_i(\xx_{-i})\big) \nonumber \\  
	&=  (1 - \STUBB)\d\big(\s_i, \aggregation_i(\xx_{-i})\big) \nonumber\\
	&\leq   (1 - \STUBB)\d\big(\x_i, \s_i\big)    + (1 - \STUBB)\d\big(\x_i, \aggregation_i(\xx_{-i})\big)\label{dominant:eq:step1}\\
	&\leq  \STUBB\d\big(\x_i, \s_i\big) + (1 - \STUBB)\d\big(\x_i, \aggregation_i(\xx_{-i})\big) \label{dominant:eq:step2} \\ 
	&=   c_i(\xx), \nonumber
	\end{align}
	where  \eqref{dominant:eq:step1} follows from the triangle inequality and \eqref{dominant:eq:step2} follows from $\STUBB \geq 1/2$, which implies that $(1 - \STUBB) \leq \STUBB$. If $\STUBB > 1/2$, then \eqref{dominant:eq:step2} is strict. Therefore, $\s_i$ is a strictly dominant strategy of agent $i$. 
\end{proof}

Theorems~\ref{thm:existence:<}~and~\ref{thm:existence:>} are consequences of lemmas~\ref{lem:BR}~and~\ref{lem:dominant-strategy}. Characterizing existence of pure Nash equilibria when $\STUBB \geq 1/2$ and the game is (semi-)restricted ($N_S \neq N$) is an interesting direction for further research. 

\begin{theorem}\label{thm:existence:<}
If $\STUBB \leq 1/2$, any consensus $\ee \in \Z^n$ is an equilibrium.  
\end{theorem}

\begin{theorem}\label{thm:existence:>}
{
If $\STUBB \geq 1/2$ (resp. $\STUBB > 1/2$) and the game is unrestricted ($N_S = N$) 
then $\ee \in \Z^n$ is an equilibrium if (resp. and only if) $\ee(i) = s_i$ for all agents $i \in N$.}
\end{theorem}

%\begin{theorem}\label{thm:existence:=}
%If $\STUBB = 1/2$ and the game is unrestricted  (i.e., $N_S = N$), then $\ss$ and any consensus belong to the set of pure Nash equilibria. 
%\end{theorem}

% We will show an interesting case in Section \ref{sec:voting}.

%Things become more involved when . In this case, in order to guarantee the existence of pure Nash equilibria, we need to make more assumptions on the structure of the game.

 %%%%%%%%%%%%%%%%%%%%%_POA_%%%%%%%%%%%%%%%%%%%%%%%%%%

\section{Price of Anarchy}\label{sec:PoA}
%We next derive bounds on the price of anarchy of general preference games with local aggregation. 

%Due to lack of space, most of the proofs of this section can be found in Section~\ref{sec:app:PoA}.  

%\subsection{Compliant Agents ($\STUBB \leq 1/2$)}
\medskip\noindent{\bf Compliant Agents.}
We first consider the case of compliant agents, where $\alpha \leq 1/2$. 
{If $\alpha = 0$, the price of anarchy is $1$}. %any consensus is an equilibrium and an optimal state, and the price of anarchy is $1$.} 
On the other hand, for $\alpha \in (0, 1/2]$, the price of anarchy is either unbounded or $+\infty$.  

\begin{proposition}\label{prop:poa:0}
If $\STUBB = 0$, $\POASUM = \POAMAX = 1$. 
\end{proposition}
\begin{proof}
Let us assume $\STUBB = 0$. 
The cost incurred by any agent $i$ in any state $\xx$ is $\d\big(\xx(i), \aggregation_i(\xx_{-i})\big)$.
By Lemma~\ref{lem:BR}, the cost incurred by any agent at any equilibrium is $0$, from which the claim immediately follows. 
\end{proof}

%\alert{
%\begin{proposition}\label{prop:poa:0}
%If $\STUBB = 0$, $\POASUM = \POAMAX = 1$. 
%\end{proposition}
%}
%\alert{
%\begin{proof}
%	Let us assume $\alpha = 0$. 
%	By Theorem \ref{thm:existence:<} and Theorem \ref{thm:existence:=} every pure Nash equilibrium is a consensus. 
%	Moreover, for every consensus $\xx$ we have $c_i(\xx) = 0$
%	for every agent $i$, which implies $\SUM(\xx) = \MAX(\xx) =0$. 
%	Since the social cost of every state is non-negative, we have that every consensus is also a social optimum (with respect to both $\SUM$ and $\MAX$), from which the claim follows. 
%\end{proof}
%}

\begin{proposition}\label{prop:poa:(0,1/2]} 
For $\STUBB \in (0,1/2]$, if the game is restricted  (${\N_\S} = \emptyset$), 
there exist instances for which both $\POASUM$ and $\POAMAX$ are unbounded. 
\end{proposition}
\begin{proof}
	Let us assume $\STUBB \in (0,1/2]$.
	Let us consider the set of instances with $\U = \{a, b, \s\}$, $\Z=\{a,b\}$ and $S = \{s\}$, (i.e., $\s_i=s$, for every $i\in\N$), where $a,b,s$ are three distinct elements.
	Notice that ${\N_\S} = \emptyset$.
	%We recall that  $0 < \STUBB \leq 1/2$. 
	%and that $\metric{\U, \d}$ is a metric space.
	Since $a,b,\s$ are three distinct elements, the distance between every two elements of $\U$ is non-negative.  
	Let the distance between $b$ and $s$ be an arbitrarily small positive number, i.e.,  $\d(b,\s) = \epsilon > 0$.
	By Theorem \ref{thm:existence:<}, the state $\ee$ in which every agent choses $a$ is a pure Nash equilibrium. 
	%In fact, for every agent $i\in \N$, $c_i\big((\ee_{-i},b)\big) = \d\big(b, \aggregation_i(\ee_{-i})\big) + \STUBB\d(b,\s) = \d(b,a) + \STUBB\d(b,\s) \geq \STUBB\big(\d(b,a) + \d(b,\s)\big) \geq \STUBB\d(a,\s) = \d(a,a) + \STUBB\d(a,\s) = \d\big(a, \aggregation_i(\ee_{-i})\big) + \STUBB\d(a,\s) = c_i(\ee)$, where the first inequality follows from the hypothesis $\STUBB \in (0,1]$, the second inequality from the triangle inequality, the second and forth equality from the fact that $\aggregation_i(\ee_{-i}) = a$ and the third equality from the fact that $\d(a,a) = 0$.
	Therefore, $\SUM(\ee) = n\STUBB\d(a, \s)$  and $\MAX(\ee) = \STUBB\d(a, \s)$.
	On the other hand, the cost of every agent $i$ in state $\xx$, in which every agent choses $b$, is  $c_i(\xx) =  \STUBB\d(b, \s) + (1-\STUBB)\d\big(b, \aggregation_i(\xx_{-i})\big)  =   \STUBB\d(b, \s) + (1-\STUBB)\d\big(b, b\big)  =  \STUBB\d(b, \s)$, where the second equality follows the fact that $\aggregation_i(\xx_{-i}) = b$  and the third from the fact that $\d(b,b) = 0$.
	Hence $\SUM(\oo) \leq \SUM(\xx) = n \STUBB\d(b, \s) = n \STUBB\epsilon > 0$ and $\MAX(\oo) \leq \MAX(\xx) =  \STUBB\d(b, \s) = \STUBB\epsilon > 0$.
	We can conclude that for this instance $\POASUM = \POAMAX \leq \frac{\d(a, \s)}{\epsilon}$.  
	Since $\epsilon$ is an arbitrarily small positive number, the claim follows.
\end{proof}

By slightly modifying the proof of Proposition~\ref{prop:poa:(0,1/2]} and by letting $b$ and $s$ coincide, i.e., $\d(b, \s) = 0$, we obtain a set of instances for the unrestricted game in which $\SUM(\xx) = \MAX(\xx) = 0$, from which the following proposition immediately follows.

%The proof of Proposition~\ref{prop:poa:(0,1/2]} considers a strategy universe $U = \{a, b, s\}$ with $Z = \{a, b\}$ and $s_i = s$ for all agents $i \in N$. If $d(b, s) = \epsilon$, for some arbitrarily small $\epsilon > 0$, $d(a, b) = 1$ and $d(a, s) = 1+\epsilon$, $d$ is a metric. Then, the consensus on $b$ is optimal, while the consensus on $a$ is an equilibrium. If, instead, we consider $Z = U = \{a, s\}$, with $d(a, s) = 1$ and $s_i = s$ for all agents $i$, the consensus on $s$ is optimal, while the consensus on $a$ is an equilibrium. Hence, we show that: 

\begin{proposition}\label{prop:poa:(0,1/2]:b}
For $\STUBB \in (0,1/2]$, if the game is unrestricted  (${\N_\S} = \N$), 
there exist instances with $\POASUM = \POAMAX=+\infty$. 
\end{proposition}

%\subsection{Stubborn Agents ($\STUBB > 1/2$)}
\noindent{\bf Stubborn Agents.}
We proceed to the case of stubborn agents, where $\alpha >  1/2$. In Theorem~\ref{thm:poa:sum}~and Theorem~\ref{thm:poa:max}, we show general bounds on the price of anarchy that depend on $\SI$, $\STR$, $\BOU$ and $\STUBB$. The proof of Lemma~\ref{lem:equilibrium} follows from the equilibrium condition, the triangle inequality and the definition of stretch. Lemma~\ref{lem:os} follows from Lemma~\ref{lem:dominant-strategy}. 

\begin{lemma}\label{lem:equilibrium} 
For every agent $i$, equilibrium $\ee$ and state $\zz$, $c_i(\ee) \leq  c_i(\zz)  + \STR_i(1-\STUBB)\ds_i(\zz, \ee)\,.$ 
%\\ \centerline{$c_i(\ee) \leq  c_i(\zz)  + \STR_i(1-\STUBB)\ds_i(\zz, \ee)\,.$}
\end{lemma}
\begin{proof} Using that $\ee$ is an equilibrium state, we have 
	\begin{align*}	
		c_i(\ee) &\leq   \STUBB\d\big(\zz(i), \s_i\big) + (1-\STUBB)\d\big(\zz(i), \aggregation_i(\ee_{-i})\big) \\
		&\leq   \STUBB\d\big(\zz(i), \s_i\big) + (1-\STUBB)\d\big(\zz(i), \aggregation_i(\zz_{-i})\big) + (1-\STUBB)\d\big(\aggregation_i(\zz_{-i}), \aggregation_i(\ee_{-i})\big) \\
		&=   c_i(\zz)  + (1-\STUBB)\d\big(\aggregation_i(\zz_{-i}), \aggregation_i(\ee_{-i})\big) \\
		&\leq    c_i(\zz)  + \STR_i(1-\STUBB)\ds_i(\zz, \ee)\,, 
	\end{align*}
	where the first inequality follows from the equilibrium condition, the second from the triangle inequality, and the last from the definition of stretch. Notably, the proof only requires that the restriction of $d$ to $Z$ satisfies the triangle inequality. 
	\end{proof}

\begin{lemma}\label{lem:os}
If $\STUBB > 1/2$, every equilibrium $\ee$ and optimal state $\oo \in \OSUM$ with $\D(\oo,\ee) \neq \emptyset$, we have $\D(\oo,\ee) \subseteq \D(\oo,\ss)$, where $\ss = (\s_i, \s_2, \ldots, \s_n)$.
\end{lemma}
\begin{proof}
Let us assume $\STUBB > 1/2$.
Let $i$ be any agent in $\D(\oo,\ee)$. 
We prove the claim by showing that $i$ belongs also to $\D(\oo,\ss)$.
If $i \in {\N_\S}$ then, by Theorem \ref{thm:existence:>}, $\ee(i) = \s_i$. 
Since $\oo(i) \neq \ee(i)$, this implies that $\oo(i) \neq \s_i$, i.e., $i\in \D(\oo,\ss)$.
If $i \in {\N\setminus \N_\S}$ then $\s_i$ does not belongs to $\Z$. 
Since $\oo(i)$ belongs to $\Z$, this trivially implies that $\oo(i) \neq \s_i$, i.e., $i\in \D(\oo,\ss)$.
%By the way of contradiction, let us assume $\D(\oo, \ss) = \emptyset$, i.e., $\o_i = \s_i$ (or equivalently $\d(\o_i,\s_i) = 0$) for every $i\in \N$. 
%This means that $\N = \N_S$ and hence, by Lemma \ref{lem:dominant-strategy}, that $\e_i \ = \s_i$ for every $i\in \N$. 
%By transitivity we obtain $\o_i \ = \e_i$ for every $i\in \N$, i.e., $\D(\oo,\ee) = \emptyset$, contradicting the hypothesis.
\end{proof}

\begin{theorem}\label{thm:poa:sum}
If $\STUBB > 1/2$, $\POASUM \leq 1 + \frac{(1-\STUBB)}{\STUBB}\frac{\SI\STR}{\BOU}$\,.
\end{theorem}
\begin{proof}
If $\E \subseteq \OSUM$, then $\POASUM = 1$. Otherwise, let $\ee$ be any equilibrium and $\oo \in \OSUM$ be any optimal state with $\D(\oo,\ee) \neq \emptyset$. We have
\begin{align}
		%&\quad  \nonumber\\
		\SUM(\ee) &\leq  \sum_{i\in \N}c_i(\oo) +  \sum_{i\in \N}\STR_i(1-\STUBB)\ds_i(\oo, \ee)\label{eq:sum:step1} \\
		&\leq  \sum_{i\in \N}c_i(\oo)  +  \STR(1-\STUBB)\sum_{i\in \N}\sum_{j\in \N\setminus\{i\}}\!\w_{ij}\d\big(\oo(j), \ee(j)\big)\label{eq:sum:step2} \\
		&=  \SUM(\oo)  +  \STR(1-\STUBB)\sum_{j\in \D(\oo,\ee)}\!\SI_{j}\d\big(\oo(j), \ee(j)\big), \label{eq:sum:step3}
	\end{align}
	where \eqref{eq:sum:step1} follows from Lemma \ref{lem:equilibrium},  \eqref{eq:sum:step2} from the definition of $\STR$ %the stretch \eqref{def:stretch} 
	and $\ds_i$,  and \eqref{eq:sum:step3}  from the definitions of $\SUM$, $\D(\oo,\ee)$ and the social impact of $j$.

On the other hand,
		\begin{align}
		%&\quad \SUM(\oo) \nonumber \\ 
	\SUM(\oo)	&\geq  
				\STUBB\sum_{
				i\in\N
				}
				\d\big(\oo(i), \s_i\big)
				\enspace = \enspace  
				\STUBB\!\sum_{
				i\in {\D(\oo,\ss)}
				}\!
				\d(\oo(i), \s_i)\nonumber\\  
				&\geq  
				\STUBB\!\sum_{
					i\in {\D(\oo,\ee)}
					}\!
					\d\big(\oo(i), \s_i\big)\label{opt:sum:step1}\\
				&\geq  
				\STUBB\!\sum_{
					i\in {\D(\oo,\ee)}
					}\!
					\BOU_i\d\big(\oo(i), \ee(i)\big),\label{opt:sum:step2}
	\end{align}
	where \eqref{opt:sum:step1} follows from Lemma \ref{lem:os} and \eqref{opt:sum:step2} from the definition of  boundary \eqref{def:boundary}.
	Notice that the last expression is strictly larger than $0$ because $\STUBB > 1/2$, $\BOU_i > 0$ and $\d\big(\oo(i), \ee(i)\big) > 0$ for every $i\in\D(\oo,\ee) \neq \emptyset$.\\
	Therefore, we can conclude that
	\begin{align}
		%&\quad \POASUM \nonumber \\
\POASUM	&\leq  
		\frac{
			\SUM(\oo) +  \STR(1-\STUBB)\!\!\!\sum\limits_{j\in \D(\oo,\ee)}\!\!\!\SI_{j}\d\big(\oo(j), \ee(j)\big)
			}{
			\SUM(\oo)
			}\label{poa:sum:step1}\\
		&\leq  1 + 
		\frac{
		\STR(1-\STUBB)\!\!\!\sum\limits_{j\in \D(\oo,\ee)}\!\!\!\SI_{j}\d\big(\oo(j), \ee(j)\big)
		}{
		\STUBB\!\!\!\sum\limits_{
					i\in {\D(\oo,\ee)}
					}\!\!\!
					\BOU_i\d\big(\oo(i), \ee(i)\big)
		}\label{poa:sum:step2}\\
		&\leq  1 + \frac{(1-\STUBB)}{\STUBB}\frac{\SI\STR}{\BOU}\label{poa:sum:step3},
	\end{align}
	where \eqref{poa:sum:step1} follows from \eqref{eq:sum:step3}, \eqref{poa:sum:step2} from \eqref{opt:sum:step2} and \eqref{poa:sum:step3} from the definitions of $\SI$ and $\BOU$.
\end{proof}

\iffalse
\begin{proof}[Proof sketch.]
If $\E \subseteq \OSUM$, then $\POASUM = 1$. Otherwise, let $\ee$ be any equilibrium and $\oo \in \OSUM$ be any optimal state with $\D(\oo,\ee) \neq \emptyset$. Then, using Lemma~\ref{lem:equilibrium} and the definitions of stretch $\tau$, $\D(\oo,\ee)$ and $\delta_j$, we get:
\[
 \SUM(\ee) \leq 
 \SUM(\oo)  +  
 \STR(1-\STUBB)\sum_{j\in \D(\oo,\ee)}\!\SI_{j}\d\big(\oo(j), \ee(j)\big)\,.
\]
	
On the other hand, using Lemma~\ref{lem:os} and the definition of  boundary \eqref{def:boundary}, we obtain that: 
\[ 
 \SUM(\oo) \geq \STUBB\!\sum_{i\in {\D(\oo,\ee)}}\!
					\BOU_i\d\big(\oo(i), \ee(i)\big)\,.
\]

The right-hand-side is larger than $0$, because $\STUBB > 1/2$, $\BOU_i > 0$ and $\d\big(\oo(i), \ee(i)\big) > 0$ for all $i\in\D(\oo,\ee) \neq \emptyset$. Taking the ratio of the two bounds above, we conclude:
\begin{align*}
   \POASUM 
		&\leq  1 + 
		\frac{
		\STR(1-\STUBB)\!\!\!\sum\limits_{j\in \D(\oo,\ee)}\!\!\!\SI_{j}\d\big(\oo(j), \ee(j)\big)
		}{
		\STUBB\!\!\!\sum\limits_{
					i\in {\D(\oo,\ee)}
					}\!\!\!
					\BOU_i\d\big(\oo(i), \ee(i)\big)
		}\\
		&\leq  1 + \frac{(1-\STUBB)}{\STUBB}\frac{\SI\STR}{\BOU},
	\end{align*}
where we also use the definitions of $\SI$ and $\BOU$.
\end{proof}
\fi

\begin{theorem}\label{thm:poa:max}
If  $\STUBB > 1/2$, $\POAMAX \leq 1 + \frac{(1-\STUBB)}{\STUBB}\frac{\STR}{\BOU}.$
\end{theorem}

\begin{proof}
If $\E \subseteq \OMAX$, then trivially $\POAMAX = 1$.

Otherwise, let  $\ee \in \E$ and $\oo \in \OMAX$  be two states such that $\D(\oo,\ee) \neq \emptyset$. 
We recall that $\STUBB > 1/2$.
Let $i$ be one of the agents with maximum cost at equilibrium, i.e., $c_i(\ee) = \MAX(\ee)$. 
We have
\begin{align}
	%&\quad\MAX(\ee) \nonumber \\
\MAX(\ee)	 &\leq  c_i(\oo) +   \STR_i(1-\STUBB)\ds_i(\oo, \ee)\label{eq:max:step1} \\
		&\leq  \MAX(\oo)  +  \STR(1-\STUBB)\!\sum_{j\in \N\setminus\{i\}}\!\w_{ij}\d\big(\oo(j),\ee(j)\big), \label{eq:max:step2}
\end{align}
	where \eqref{eq:max:step1} follows from Lemma \ref{lem:equilibrium} and \eqref{eq:max:step2}  from the definition of $\STR$ and $\ds_i$.
	Notice that, if none of the agents in  $\N\setminus\{i\}$ is in $\D(\oo,\ee)$, then every term in the summation in \eqref{eq:max:step2} is $0$ and the inequality becomes $\MAX(\ee) \leq \MAX(\oo)$, which implies $\POAMAX = 1$.
	Therefore, let us assume that $\{\N\setminus\{i\}\} \cap \D(\oo,\ee) \neq \emptyset$.
	%the existence of a agent in $\D(\oo,\ee)$ which is different from $i$.
	In particular, let  $j^*$ be any agent in such set %$\D(\oo,\ee)$, different from $i$,
	maximizing the distance between the strategy at the optimum and the one at the equilibrium, i.e.,  $j^* \in \arg\max_{\begin{subarray}{l} j \in \D(\oo, \ee)\\ j\neq i\end{subarray}}\d\big(\oo(j), \ee(j)\big)$.
	Then we can write
\begin{align}
		%&\quad\MAX(\ee) \nonumber \\ 
	\MAX(\ee)	&\leq  \MAX(\oo) + \STR(1-\STUBB)\!\sum_{j\in \N\setminus\{i\}}\!\w_{ij}\d\big(\oo({j^*}),\ee({j^*})\big)\label{eq:max:step3} \\
		& =  \MAX(\oo) + \STR(1-\STUBB)\d\big(\oo({j^*}),\ee({j^*})\big),\label{eq:max:step4}
\end{align}
where \eqref{eq:max:step3} follows from \eqref{eq:max:step2} and the definition of $j^*$ and  \eqref{eq:max:step4} from the fact that $\sum_{j\in \N} \w_{ij} = 1$.\\
On the other side,
\begin{equation}\label{opt:max:step1}
	\begin{split}
		\MAX(\oo) 
		&\geq  c_{j^*}(\oo)  \geq  \STUBB\d\big(\oo({j^*}), \s_{j^*}\big) \\
		&\geq  \STUBB\BOU_i\cdot\d\big(\oo({j^*}), \ee({j^*})\big), 
	\end{split}
\end{equation}	
	where the last inequality follows from  the definition of  boundary \eqref{def:boundary}. 
		Notice that the last expression is strictly larger than $0$ because $\STUBB > 1/2$, $\BOU_i > 0$ and $\d\big(\oo({j^*}), \ee({j^*})\big) > 0$.\\	
	Therefore, we can conclude that 
	\begin{align}
		%&\quad\POAMAX\nonumber \\ 
	\POAMAX	&\leq  
		\frac{
		\MAX(\oo) + \STR(1-\STUBB)\d\big(\oo({j^*}),\ee({j^*})\big)
		}{
		\MAX(\oo)
		}\label{poa:max:step1}\\
		&\leq  
		1 + 
		\frac{
		 \STR(1-\STUBB)\d\big(\oo({j^*}),\ee({j^*})\big)
		}{
		\STUBB\BOU_i\cdot\d\big(\oo({j^*}), \ee({j^*})\big)
		} \nonumber \\
		&\leq  1 + \frac{(1-\STUBB)}{\STUBB}\frac{\STR}{\BOU},\label{poa:max:step2}
	\end{align}
	where \eqref{poa:max:step1} follows from \eqref{eq:max:step4}, and \eqref{poa:max:step2} from \eqref{opt:max:step1} and the definition of $\BOU$.
	\end{proof}

The proof of Theorem~\ref{thm:poa:max} is similar to that of Theorem~\ref{thm:poa:sum}. We note that our price of anarchy analysis is general and the bounds of theorems~\ref{thm:poa:sum}~and~\ref{thm:poa:max} directly apply to discrete preference games that satisfy our assumptions on the influence weights $w_{ij}$\,. In this case, $\tau = 1$, because discrete preference games do not use preference aggregation. Assuming that $\alpha \in (1/2, 1]$ and $\sum_{j} w_{ij} = 1$ (but note that the latter does not hold in \cite{ChierichettiKO18}), $\POASUM \leq \delta(2\alpha-1)^{-1}$ and $\POAMAX \leq (2\alpha-1)^{-1}$ (where we use also that $\beta \geq \frac{2\alpha-1}{2\alpha}$, Theorem~\ref{thm:boundary:Umetric}). Moreover, in fully symmetric discrete preference games, where $w_{ij} = w_{ji}$ and $\delta = 1$, we get that $\POASUM \leq (2\alpha - 1)^{-1}$, for all $\alpha \in 
(1/2, 1]$.

%%%%%%%%%%%%%%%%%%%%%_GENERAL_BOUNDS_%%%%%%%%%%%%%%%%%%%%%%%%%%

%\section{General Bounds on {$\BOU$} and {$\STR$}}
\section{Bounds on Boundary and Stretch}
\label{sec:tau_beta}

Next, we show bounds on the boundary and the stretch under general assumptions on the structure of the game. 

\begin{theorem}\label{thm:boundary:Umetric}
If $\STUBB > 1/2$ and $\d$ is an exact metric then $\BOU \geq \frac{2\STUBB-1}{2\STUBB}$.
\end{theorem}
\begin{proof} 
For any equilibrium $\ee\in \E$, any agent $i$, and any strategy $\x \neq \ee(i)$ (recall that $|\Z| \geq 2$), we apply the equilibrium condition and obtain that 
	\begin{equation}\label{lb1_g:br}
		\STUBB\d\big(\ee(i), \s_i\big) + (1-\STUBB)\d\big(\ee(i), \aggregation_i(\ee_{-i})\big) \leq  \STUBB\d(\x, \s_i) + (1-\STUBB)\d\big(\x, \aggregation_i(\ee_{-i})\big)\,.
	\end{equation}
	
	Therefore,
	\begin{align*}
%		&\quad\d(\x, \s_i) \\
		 \d(\x, \s_i)  &\geq   \d\big(\ee(i), \s_i\big) - \frac{(1-\STUBB)}{\STUBB}\Big[\d\big(\x, \aggregation_i(\ee_{-i})\big)  - \d\big(\ee(i), \aggregation_i(\ee_{-i})\big)\Big] \\
		&\geq  \d\big(\ee(i), \s_i\big) - \frac{(1-\STUBB)}{\STUBB}\d\big(\x,\ee(i)\big) \\
		&\geq  \Big[\d\big(\x,\ee(i)\big) - \d(\x,\s_i)\Big] - \frac{(1-\STUBB)}{\STUBB}\d\big(\x,\ee(i)\big) \\
		&\geq   \left(\frac{2\STUBB - 1}{\STUBB}\right)\d\big(\x,\ee(i)\big)  - \d(\x,\s_i),
	\end{align*}
	where the first inequality follows from \eqref{lb1_g:br} and the remaining inequalities from the triangle inequality. The theorem follows by adding $\d(\x,\s_i)$ to each side of the previous inequality and dividing by $2$. 
	\end{proof}

\begin{proposition}\label{prop:boundary:Ns=N}
For $\STUBB > 1/2$, if either the game is unrestricted ($\N_\S = \N$) or {$\d$ is uniform} then $\BOU = 1$.
\end{proposition}

\begin{proof}
Let us first assume ${\N_\S} = \N$. 
Let us consider any equilibrium $\ee\in \E$, any agent $i$ and any strategy $\x \neq \ee(i)$ (recall that $|\Z| \geq 2$).
In order to prove the claim, we need to show that $\d(\x,\s_i) = \d\big(\x,\ee(i)\big)$.
If $\STUBB > 1/2$, by Theorem \ref{thm:existence:>}, we get $\ee(i) = \s_i$, which trivially implies $\d(\x, \s_i) = \d\big(\x, \ee(i)\big)$.

Now, let us assume that $\d$ is uniform. 
Let us consider any equilibrium $\ee\in \E$, any agent $i$ and any strategy $\x \neq \ee(i)$ (recall that $|\Z| \geq 2$).
In order to prove the claim, we need to show that $\d(\x,\s_i) = \d\big(\x,\ee(i)\big)$.
Notice that, since $\x \neq \ee(i)$, we have $\d\big(\x,\ee(i)\big) = 1$.
Therefore, in order for the claim to hold, it must be that $\d(\x,\s_i) = 1$.
By contradiction, let us assume that $\d(\x,\s_i) = 0$, or equivalently $\x = \s_i$. 
This implies that $i$ belongs to $\N_\S$. 
If $\STUBB > 1/2$, by Theorem \ref{thm:existence:>}, we get $\s_i = \ee(i)$, and by transitivity $\x = \ee(i)$, i.e., $\d\big(\x,\ee(i)\big) = 0$, contradicting the hypothesis.
\end{proof}

%\begin{proposition}\label{prop:boundary:uniform}
%For all $\STUBB \in (1/2, 1]$,  if $\metric{\U, \d}$ is the uniform metric space, then $\BOU = 1$.
%\end{proposition}
%\begin{proof}
%Let us consider any equilibrium $\ee\in \E$, any agent $i$ and any strategy $\x \neq \ee(i)$ (recall that $|\Z| \geq 2$).
%In order to prove the claim, we need to show that $\d(\x,\s_i) = \d\big(\x,\ee(i)\big)$.
%Notice that, since $\x \neq \ee(i)$, we have $\d\big(\x,\ee(i)\big) = 1$.
%Therefore, in order for the claim to hold, it must be that $\d(\x,\s_i) = 1$.
%By contradiction, let us assume that $\d(\x,\s_i) = 0$, or equivalently $\x = \s_i$. 
%This implies that $i$ belongs to $\N_\S$. 
%If we assume $\STUBB > 1/2$, by Theorem \ref{thm:existence:>}, we get $\s_i = \ee(i)$, and by transitivity $\x = \ee(i)$, i.e., $\d\big(\x,\ee(i)\big) = 0$, contradicting the hypothesis.
%\end{proof}

\begin{proposition}\label{prop:general_tau}
The global stretch $\tau \leq \frac{d^{\max}(Z)}{w^{\min} d^{\min}(Z)}$, where $d^{\max}(Z) = \max_{x, y \in Z} d(x, y)$ is the diameter of $Z$, $d^{\min}(Z) = \min_{x \neq y: d(x, y) > 0} d(x, y)$ is the minimum positive distance in $Z$, and $w^{\min} = \min_{i \neq j: w_{ij} > 0} w_{ij}$ is the minimum positive influence weight. 
\end{proposition}
\begin{proof}
In the definition of $\tau_i$, we have $\tau_i = 1$, if $\ds_i(\vec{e}, \vec{y}) = d(\aggr_i(\vec{e}_{-i}), \aggr_i(\vec{y}_{-i})) = 0$. Hence, we assume that $\ds_i(\vec{e}, \vec{y}) > 0$. 

We let $w^{\min}(i) = \min_{j \neq i: w_{ij} > 0} w_{ij}$\,. Then, $\ds_i(\vec{e}, \vec{y}) = \sum_{j \neq i} w_{ij} d(\vec{e}(j), \vec{y}(j)) \geq w^{\min}(i)\, d^{\min}(Z)$, because $\vec{e}(j), \vec{y}(j) \in Z$, there is at least one positive term in the sum, and if either $w_{ij} = 0$ or $d(\vec{e}(j), \vec{y}(j)) = 0$, the corresponding term  is $0$. Moreover, $d(\aggr_i(\vec{e}_{-i}), \aggr_i(\vec{y}_{-i})) \leq d^{\max}(Z)$, because $\aggr_i(\vec{e}_{-i}), \aggr_i(\vec{y}_{-i}) \in Z$. Therefore, 
\[ \tau_i \leq \frac{d^{\max}(Z)}{w^{\min}(i)\,d^{\min}(Z)}. \]
Using that the global stretch $\tau = \max_{i \in N} \tau_i$ and that $w^{\min} = \min_{i \in N} w^{\min}(i)$, we conclude the proof of the proposition. 
\end{proof}

Unless we impose additional structure, the upper bound of Proposition~\ref{prop:general_tau} is essentially the best possible, because there are examples where a small change in a single coordinate of a state moves the aggregate to a diametrically different strategy. E.g., consider $Z = \{ 0, 1 \}$, $\vec{e}_{-i} = (0, 1, 1)$, $\vec{y}_{-i} = (0, 0, 1)$, $w_{i1} = w_{i3} = \frac{1-\varepsilon}{2}$ and $w_{i2} = \varepsilon$, and the Fr\'{e}chet median as aggregation function. Clearly, $\aggr_i(\vec{e}_{-i}) = 1$,  while $\aggr_i(\vec{y}_{-i}) = 0$. 

For the following propositions, we restrict the definition of $\STR$ to optimal states $\vec{o}$ (either $\vec{o} \in \OSUM$ or $\vec{o} \in \OMAX$), instead of arbitrary states $\vec{y}$.

\begin{proposition}\label{prop:strech:o-consensus}
If every social optimum (with respect to any social objective) is a consensus, $\d$ is an exact metric on $\Z$ and the aggregation function is the Fr\'{e}chet median then $\STR \leq 2$.  
\end{proposition}
\begin{proof}
%We start with the proof of Proposition~\ref{prop:strech:o-consensus}. 
Let us consider any equilibrium $\ee\in \E$, any social optimum $\oo$ (either $\vec{o} \in \OSUM$ or $\vec{o} \in \OMAX$) and any agent $i$. 

%First notice that, since $\aggregation_{i}(\ee_{-i}) \in  {\arg\min}_{\y\in \Z} \sum_{j\in \N\setminus\{i\}}\w_{ij}\d\big(\y, \ee(j)\big)$, it holds  
By the definition of the Fr\'{e}chet mean, 
\begin{equation}\label{eq:tau:init}
		\sum_{j\in \N\setminus\{i\}}\w_{ij}\d\big(\aggregation_{i}(\ee_{-i}), \ee(j)\big) \leq \sum_{j\in \N\setminus\{i\}} \w_{ij}\d\big(\aggregation_{i}(\oo_{-i}), \ee(j)\big).
\end{equation}

Hence, we have  that
	\begin{align}
 \d\big(\aggregation_{i}(\oo_{-i}),\aggregation_{i}(\ee_{-i})\big)		 & =   \sum_{j\in \N\setminus\{i\}} \w_{ij} \d\big(\aggregation_{i}(\oo_{-i}),\aggregation_{i}(\ee_{-i})\big) \label{eq:tau:step1} \\ 
		 & \leq  \sum_{j\in \N\setminus\{i\}} \Big[\w_{ij}\d\big(\aggregation_{i}(\oo_{-i}), \ee(j)\big)  + \w_{ij}\d\big(\ee(j), \aggregation_{i}(\ee_{-i})\big) \Big] \label{eq:tau:step2}\\
		 & =  \sum_{j\in \N\setminus\{i\}} \w_{ij}\d\big(\aggregation_{i}(\oo_{-i}), \ee(j)\big)  + \sum_{j\in \N\setminus\{i\}}\w_{ij}\d\big(\aggregation_{i}(\ee_{-i}), \ee(j)\big) \nonumber \\	 
		 & \leq  2\sum_{j\in \N\setminus\{i\}} \w_{ij}\d\big(\aggregation_{i}(\oo_{-i}), \ee(j)\big) \label{eq:tau:step3}\\
		 & =  2\sum_{j\in \N\setminus\{i\}}\w_{ij}\d\big(\oo(j),\ee(j)\big) = 2\cdot\ds_i(\oo, \ee), \label{eq:tau:step4}
	\end{align}
	where \eqref{eq:tau:step1} follows from $\sum_{j\in \N\setminus\{i\}}\w_{ij} = 1$, \eqref{eq:tau:step2} from the triangle inequality, \eqref{eq:tau:step3} from   \eqref{eq:tau:init} and \eqref{eq:tau:step4} from the fact that $\oo$ is a consensus.
	%The same argument will apply if we take $\oo \in \OMAX$. 
\end{proof}

The proof of Proposition~\ref{prop:strech:e-consensus} is symmetric to that of Proposition~\ref{prop:strech:o-consensus} and follows if we simply exchange the roles of $\ee$ and $\oo$ in the argument of the proof.

\begin{proposition}\label{prop:strech:e-consensus}
If every equilibrium is a consensus, $\d$ is an exact metric on $\Z$  and the aggregation function is the Fr\'{e}chet median then $\STR \leq 2$.
\end{proposition}

\section{$k$-Approval Voting Game}
\label{sec:voting}

Next, we introduce a natural preference game with local aggregation,  
motivated by $k$-approval voting with $m > k$ candidates. The strategy space $Z$ consists of all binary strings of length $m$ with $k$ ones. Formally, $Z = \{ \vec{u} \in \{0, 1\}^m\,|\, \sum_{\ell=1}^m \vec{u}(\ell) = k \}$. The strategy universe $U = \{ \vec{u} \in [0, 1]^m\,|\, \sum_{\ell=1}^m \vec{u}(\ell) = k \}$. An agent $i$ can choose any $s_i \in U$ as her preferred strategy. The distance function is $L_2^2$, which is a $2$-approximate metric. Namely, for every $\vec{u}, \vec{v} \in U$, 
\[ d(\vec{u}, \vec{v}) = \sum_{\ell=1}^m (\vec{u}(\ell) - \vec{v}(\ell))^2\,. \]
Note that $L_2^2$ is identical to the Hamming distance (which is a metric distance) when restricted to $Z$. The aggregation function is the Fr\'{e}chet median in $Z$. 

Our approval voting model does not admit a potential function (due to the asymmetry of $w_{ij}$, even for $m = 2$ and $k=1$). Using \cite[Observation~2.2]{FerraioliGV16}, we can show that our approval voting game admits a pure Nash equilibrium for $m = 2$, $k=1$ and $\alpha \in [0, 1]$. An interesting open problem is to characterize the values of $k$ (potentially as a function of $m$) for which the $k$-approval voting game admits pure Nash equilibria. 

The following theorem, derives a nontrivial lower bound on $\beta$ for the $k$-approval voting game. 

\begin{theorem}\label{thm:boundary:voting}
For every integers $k < m$, if $\STUBB > 1/2$ then $\BOU \geq \left(\frac{2\STUBB-1}{2\STUBB}\right)^2$ for the $k$-approval voting game.% with local aggregation.
\end{theorem}
\begin{proof}
To obtain a lower bound on $\BOU_i$, we fix an equilibrium $\ee \in \E$ and a strategy $\x \in \Z$ with $\ee(i) \neq \x$, and find $\s^\ast_i = \arg\min_{\s_i \in U} \d(\x, \s_i) / \d(\x, \e_i)$. Then, we conclude that $\BOU_i \geq \d(\x, \s^\ast_i) / \d(\x, \ee(i))$. Throughout the proof, we let $[m] = \{ 1, \ldots, m\}$. Moreover, we let $x(j)$, $s_i(j)$ and $\ee(i)(j)$ denote the $j$-th coordinate (bit) of strategies $x$, $s_i$ and $\ee(i)$. 

To this end, for any fixed strategy $\x \in \Z$, let $C_0 = \{ j \in [m] : \x(j) = 0 \}$ and let $C_1 = \{ j \in [m] : \x(j) = 1 \}$ be indices $j$ for which $\x(j)$ is $0$, for $C_0$, and $1$, for $C_1$. We recall that $|C_1| = k$ and $|C_0| = m-k$. Moreover, for the equilibrium strategy $\ee(i) \in \Z$, we let $C_{01} = \{ j \in C_0 : \ee(i)(j) = 1 \}$ and $C_{10} = \{ j \in C_1 : \ee(i)(j) = 0 \}$ be indices $j$ where $\x(j) \neq \ee(i)(j)$. Since both $\x$ and $\ee(i)$ include exactly $k$ ones and $m-k$ zeros, $|C_{01}| = |C_{10}|$. In the following, we denote $|C_{01}| = |C_{10}| = \ell$, for simplicity, for some $1 \leq \ell \leq k$. Using this notation, we can write that $\d(\x, \ee(i)) = 2\ell$. 

Our goal is to find $\min_{\s_i \in U} \d(\x, \s_i) / (2\ell)$, for any fixed $\ell \in \{ 1, \ldots, k \}$, under the constraints that $\s_i \in U$, $\x \in Z$, and $\ee(i) \in Z$ is an equilibrium strategy. We first observe that 
\begin{equation}\label{eq:man_objective}
\frac{\d(\x, \s_i)}{2\ell} = \frac{1}{2\ell}\left(\sum_{j \in C_0} \s_i(j)^2 + \sum_{j \in C_1} (1-\s_i(j))^2 \right), 
\end{equation}
where 
\begin{equation}\label{eq:man_constr1}
\sum_{j \in C_0 \cup C_1} \s_i(j) = k \mbox{\ \ and \ \ } \s_i \in [0, 1]^m.
\end{equation}

Since $\ee$ is an equilibrium, and following the same reasoning as in the proof of Theorem~\ref{thm:boundary:Umetric}, we obtain that
\begin{align}
 &\quad \d(\ee(i), \s_i) - \d(\x, \s_i) \notag \\
 & \leq (1-\STUBB)\left( \d\big(\x, \aggregation_i(\ee_{-i})\big) - 
 \d\big(\ee(i), \aggregation_i(\ee_{-i})\big) \right) / \STUBB \notag \\
& \leq (1-\STUBB)\d(\x, \ee(i)) / \STUBB = 2\ell(1-\STUBB) / \STUBB\,, \label{eq:man_c1} 
\end{align}
where the first inequality follows from \eqref{lb1_g:br}, in the proof of Theorem~\ref{thm:boundary:Umetric}, and the second inequality follows from the triangle inequality, the fact that $\x, \ee(i), \aggregation_i(\ee_{-i}) \in \Z$, and the observation that the triangle inequality holds for $\d$ on $\Z$, since $L_2^2$ is identical to the Hamming distance on $Z$.

Since for any $j \not\in C_{10} \cup C_{01}$, $\x(j) = \ee(i)(j)$, we have that
\begin{align*}
  &\quad \d(\ee(i), \s_i) - \d(\x, \s_i) \\
  & = \sum_{j \in C_{10}} \left( \s_i(j)^2 - (1- \s_i(j))^2\right) \\
  & \ \ \ \ \ + \sum_{j \in C_{01}} \left( (1-\s_i(j))^2 - \s_i(j)^2\right) \\
 & = 2 \sum_{j \in C_{10}} \s_i(j) - 2 \sum_{j \in C_{01}} \s_i(j) \,,
\end{align*}
where for the last equality, we use the fact that $|C_{10}| = |C_{10}|$. Using \eqref{eq:man_c1}, we conclude that for any equilibrium strategy $\ee(i) \in Z$ and any strategy $\x \in Z$ with $d(\x, \ee(i)) = 2\ell$, $\s_i$ satisfies
\begin{equation}\label{eq:man_constr2}
  \sum_{j \in C_{10}} \s_i(j) - \sum_{j \in C_{01}} \s_i(j) \leq \ell (1-\STUBB)/ \STUBB.
\end{equation}

Therefore, we can find $\min_{\s_i \in S} \d(\x, \s_i) / (2\ell)$ by computing an $\s_i$ that minimizes~\eqref{eq:man_objective}, subject to~\eqref{eq:man_constr1}~and~\eqref{eq:man_constr2}. Since \eqref{eq:man_objective} is a strictly convex function of $\s_i(1), \ldots, \s_i(m)$ and the constraints of \eqref{eq:man_constr1} and \eqref{eq:man_constr2} are linear on $\s_i(1), \ldots, \s_i(m)$, we have to solve a strictly convex minimization problem with linear constraints. 

We observe that 
\[ \s^\ast_i(j) = \begin{cases}
             \displaystyle \frac{2\STUBB-1}{2\STUBB}\ \ \ \ \ &\mbox{if $j \in C_{01}$}\\ 
             \displaystyle \frac{1}{2\STUBB}          &\mbox{if $j \in C_{10}$}\\ 
             0                                 &\mbox{if $j \in C_0 \setminus C_{01}$}\\
             1                                 &\mbox{if $j \in C_1 \setminus C_{10}$} 
             \end{cases}
\]
is feasible (i.e., it lies in $\S$ and satisfies \eqref{eq:man_constr2}) and 
achieves an objective value of $\d(\x, \s^\ast_i) / (2\ell) = \left(\frac{2\STUBB-1}{2\STUBB}\right)^2$. Moreover, $\s^\ast_i$ is defined for any strategy $\x \in \Z$ and any equilibrium strategy $\ee(i) \in \Z$ with $\d(\x, \ee(i)) = 2\ell > 0$ and always attains the same objective value $\d(\x, \s^\ast_i) / (2\ell) = \left(\frac{2\STUBB-1}{2\STUBB}\right)^2$. 

To confirm that $\s^\ast_i$ is indeed a minimizer of the above strictly convex minimization problem, we observe that the KKT optimality conditions are satisfied at $\s^\ast_i$ with KKT multipliers $\mu = \frac{2\STUBB-1}{2\ell\STUBB}$ for \eqref{eq:man_constr2} and $0$ for all the constraints of \eqref{eq:man_constr1}. To verify this claim, we let 
\begin{align*}
 f(\y) &= \frac{1}{2\ell}\left(\sum_{j \in C_0} \y(j)^2 + \sum_{j \in C_1} (1-\y(j))^2 \right) \mbox{\ \ \ \ and}\\
 g(\y) &= \sum_{j \in C_{10}} \y(j) - \sum_{j \in C_{01}} \y(j) - \ell (1-\STUBB)/ \STUBB 
\end{align*}
We first observe that complementary slackness constraints are satisfied, since $g(\s^\ast_i) = 0$ and all other KKT multipliers are $0$. Primal and dual feasibility constraints are also satisfied. Regarding stationarity constraints, we observe that $\nabla_j f(\s^\ast_i) = 0$, for any $j \not\in C_{01} \cup C_{10}$, $\nabla_j f(\s^\ast_i) = \frac{2\STUBB-1}{2\ell\STUBB}$, for all $j \in C_{01}$, and $\nabla_j f(\s^\ast_i) = \frac{1-2\STUBB}{2\ell\STUBB}$, for all $j \in C_{10}$. Similarly, $\nabla_j g(\s^\ast_i) = 0$ for any $j \not\in C_{01} \cup C_{10}$, $\nabla_j g(\s^\ast_i) = -1$, for all $j \in C_{01}$, and $\nabla_j g(\s^\ast_i) = 1$, for all $j \in C_{10}$. Therefore, stationarity constraints are also satisfied with KKT multiplier $\mu = \frac{2\STUBB-1}{2\ell\STUBB}$.
\end{proof}

{
We note that the proof of Theorem~\ref{thm:boundary:voting} shows something stronger, namely that for any equilibrium strategy $\ee(i) \in \Z$ and any strategy $\x \in \Z$, with $\x \neq \ee(i)$, there is a preferred srtategy $\s^\ast_i \in S$, defined as in the proof of Theorem~\ref{thm:boundary:voting},
such that $\d(\x, \s^\ast_i) / d(\x, \ee(i)) = \left(\frac{2\STUBB-1}{2\STUBB}\right)^2$. Therefore, the proof of Theorem~\ref{thm:boundary:voting} shows that for any agent $i$, $\BOU_i = \left(\frac{2\STUBB-1}{2\STUBB}\right)^2$. 
}

Combined with Theorem~\ref{thm:poa:sum}, Theorem~\ref{thm:boundary:voting} implies that for opinion formation with binary preferences (where we do not have local aggregation and Lemma~\ref{lem:equilibrium} holds with $\tau = 1$, as an immediate consequence of the model defitinion in \cite{FerraioliGV16}), the price of anarchy is at most $\delta (2\alpha-1)^{-2}$, and at most $(2\alpha-1)^{-2}$ in the fully symmetric case where $\delta = 1$. Combining this bound with the fact that the price of anarchy is unbounded for all $\alpha \in [0, 1/2]$ \cite{FerraioliGV16}, we get the complete picture of the price of anarchy for opinion formation with binary preferences, for all $\alpha \in [0, 1]$.

%%%%%%%%%%%%%%%%%%%%%_BIB_%%%%%%%%%%%%%%%%%%%%%%%%%%
%\newpage

\bibliographystyle{alpha}
\bibliography{bib}

\newcommand{\etalchar}[1]{$^{#1}$}
\begin{thebibliography}{ABE{\etalchar{+}}18}

\bibitem[ABE{\etalchar{+}}18]{AnshelevichBEPS18}
E.~Anshelevich, O.~Bhardwaj, E.~Elkind, J.~Postl, and P.~Skowron.
\newblock Approximating optimal social choice under metric preferences.
\newblock {\em Artificial Intelligence}, 264:27--51, 2018.

\bibitem[ACF{\etalchar{+}}16]{AulettaCFGP16}
V.~Auletta, I.~Caragiannis, D.~Ferraioli, C.~Galdi, and G.~Persiano.
\newblock Generalized discrete preference games.
\newblock In {\em Proceedings of the 25th International Joint Conference on
  Artificial Intelligence (IJCAI)}, pages 53--59, 2016.

\bibitem[ACF{\etalchar{+}}17]{AulettaCFGP17}
V.~Auletta, I.~Caragiannis, D.~Ferraioli, C.~Galdi, and G.~Persiano.
\newblock Robustness in discrete preference games.
\newblock In {\em In Proc. of the 16th International Joint Conference on
  Autonomous Agents and Multiagent Systems (AAMAS)}, pages 1314--1321, 2017.

\bibitem[ACN08]{AilonCN08}
N.~Ailon, M.~Charikar, and A.~Newman.
\newblock Aggregating inconsistent information: Ranking and clustering.
\newblock {\em Journal of the {ACM}}, 55(5), 2008.

\bibitem[BFM18]{BiloFM18}
V.~Bil{\`{o}}, A.~Fanelli, and L.~Moscardelli.
\newblock Opinion formation games with dynamic social influences.
\newblock {\em Theoretical Computer Science}, 746:73--87, 2018.

\bibitem[BGM13]{BGM13}
K.~Bhawalkar, S.~Gollapudi, and K.~Munagala.
\newblock Coevolutionary opinion formation games.
\newblock In {\em In Proc. of the 45th {ACM} Symposium on Theory of Computing
  (STOC)}, pages 41--50, 2013.

\bibitem[BKO15]{BindelKO15}
D.~Bindel, J.~M. Kleinberg, and S.~Oren.
\newblock How bad is forming your own opinion?
\newblock {\em Games and Economic Behavior}, 92:248--265, 2015.

\bibitem[CKO18]{ChierichettiKO18}
F.~Chierichetti, J.~M. Kleinberg, and S.~Oren.
\newblock On discrete preferences and coordination.
\newblock {\em Journal of Computer and System Sciences}, 93:11--29, 2018.

\bibitem[Deg74]{Degroot}
M.~H. Degroot.
\newblock Reaching a consensus.
\newblock {\em Journal of the American Statistical Association},
  69(345):118--121, 1974.

\bibitem[EFHS19]{EpitropouFHS19}
M.~Epitropou, D.~Fotakis, M.~Hoefer, and S.~Skoulakis.
\newblock Opinion formation games with aggregation and negative influence.
\newblock {\em Theory Computing Systems}, 63(7):1531--1553, 2019.

\bibitem[FGV16]{FerraioliGV16}
D.~Ferraioli, P.~W. Goldberg, and C.~Ventre.
\newblock Decentralized dynamics for finite opinion games.
\newblock {\em Theoretical Computer Science}, 648:96--115, 2016.

\bibitem[FJ90]{Friedkin}
N.~E. Friedkin and E.~C. Johnsen.
\newblock Social influence and opinions.
\newblock {\em The Journal of Mathematical Sociology}, 15(3-4):193--206, 1990.

\bibitem[FKKS18]{FKKS18}
D.~Fotakis, V.~Kandiros, V.~Kontonis, and S.~Skoulakis.
\newblock Opinion dynamics with limited information.
\newblock In {\em In Proc. of the 14th International Conference on Web and
  Internet Economics (WINE)}, volume 11316 of {\em LNCS}, pages 282--296, 2018.

\bibitem[FPS16]{FPS16}
D.~Fotakis, D.~Palyvos{-}Giannas, and S.~Skoulakis.
\newblock Opinion dynamics with local interactions.
\newblock In {\em In Proc. of the 25th International Joint Conference on
  Artificial Intelligence (IJCAI)}, pages 279--285, 2016.

\bibitem[GS14]{GS12}
J.~Ghaderi and R.~Srikant.
\newblock Opinion dynamics in social networks with stubborn agents: Equilibrium
  and convergence rate.
\newblock {\em Automatica}, 50:3209--3215, 2014.

\bibitem[LBN{\etalchar{+}}19]{LolakapuriBNPD19}
P.~R. Lolakapuri, U.~Bhaskar, R.~Narayanam, G.~R. Parija, and P.~S. Dayama.
\newblock Computational aspects of equilibria in discrete preference games.
\newblock In {\em In Proc. of the 26th International Joint Conference on
  Artificial Intelligence (IJCAI)}, pages 471--477, 2019.

\bibitem[UL04]{UL04}
D.~Urbig and J.~Lorenz.
\newblock Communication regimes in opinion dynamics: Changing the number of
  communicating agents.
\newblock In {\em In Proc. of the 2nd Conference of the European Social
  Simulation Association (ESSA)}, 2004.

\end{thebibliography}

\end{document}